\documentclass[11pt]{article}
\usepackage{amsmath,amssymb, amsthm,multicol,tikz,subfigure}
\usepackage[margin=2.5cm]{geometry}
\usepackage{graphicx}
\usepackage[all]{xy}

\theoremstyle{plain}
\newtheorem{theorem}{Theorem}[section]
\newtheorem{cor}[theorem]{Corollary}
\newtheorem{prop}[theorem]{Proposition}
\newtheorem{thm}[theorem]{Theorem}
\newtheorem{lem}[theorem]{Lemma}

\newtheorem*{thm*}{Theorem}

\theoremstyle{definition}

\newcommand{\R}{\mathbb{R}}

\newcommand{\supp}{\operatorname{supp}}

\newcommand{\M}{\mathcal{M}}
\renewcommand{\int}{\operatorname{int}}

\author{Katharine Turner}

\title{Cone fields and topological sampling in manifolds with bounded curvature}
\begin{document}
\maketitle
\begin{abstract}
A standard reconstruction problem is how to discover a compact set from a noisy point cloud that approximates it. A finite point cloud is a compact set. This paper proves a reconstruction theorem which gives a sufficient condition, as a bound on the Hausdorff distance between two compact sets, for when certain offsets of these two sets are homotopic in terms of the absence of $\mu$-critical points in an annular region. We reduce the problem of reconstructing a subset from a point cloud to the existence of a deformation retraction from the offset of the subset to the subset itself. The ambient space can be any Riemannian manifold but we focus on ambient manifolds which have nowhere negative curvature (this includes Euclidean space). We get an improvement on previous bounds for the case where the ambient space is Euclidean whenever $\mu\leq0.945$ ($\mu\in (0,1)$ by definition). In the process, we prove stability theorems for $\mu$-critical points when the ambient space is a manifold. 
%
\end{abstract}


\section{Introduction}
In modern science and engineering a common problem is understanding some shape from a point cloud sampled from that shape. This point cloud should be thought of as some finite number of (potentially) noisy samples. Topology and geometry are considered very natural tools in such data analysis (see e.g. \cite{ghrist2008barcodes} and \cite{carlsson2009topology}). One reason is because topological invariants are often more stable under noise.  We will want to understand the homotopy type of a set - two objects are homotopy equivalent if there is a way to continuously deform one object into another.

Often we wish to know the extent to which we can, and how to, reconstruct shapes from noisy point clouds. Naturally the more restrictions on the space the easier it is to reconstruct. The first area of focus was the study of surfaces in $\R^3$, motivated by problems such as medical imaging, visualization and reverse engineering of physical objects. Algorithms with theoretical guaranties exist for smooth closed surfaces with sufficient dense samples. In \cite{amenta2000} the concept of a local feature size   was introduced. The  local feature size at $p$, denoted $\text{lfs}(p)$, is the distance from $p$ to the medial axis of $A$.  The sampling conditions for surface reconstruction were based on the concept of $\epsilon$-sampling. A point cloud of $A$ is an $\epsilon$-sample if for every point $p \in A$ there is some sample point at distance at most $\epsilon \text{lfs}(p)$ away. The Cocone Algorithm produces a homeomorphic set from any $0.06$-sampling of a smooth closed surface \cite{amenta2000}. This process has been extended to smooth surfaces with boundaries \cite{dey2009isotopic}. However given an arbitrary compact set $K$, the best we can hope for is to find some nearby set that  is homotopy equivalent to an offset of $K$. For instance from a point cloud we can not tell apart the original set and the same set  with a slight thickening in places.

One of the simplest methods of reconstruction is to use the offset of a sampling. Given a set $K$, the $r$-offset of $K$, denoted $K_r$, is defined to be $\{x\in \M:d(x, K)\leq r\}$. This is topologically the same as taking the $\alpha$-shape of data points \cite{alpha1992} or taking the C\v ech complex \cite{chazal2008}. This leads to the problem of finding theoretical guarantees as to when an offset of a sampling has the same topology (i.e. homotopy type) as the underlying set. In other words we want to find conditions on a point cloud $S$ of a compact set $A$ so that $S_r$ is homotopy equivalent to $A$. We will in fact find sufficient conditions for $S_r$ to deformation retract to $A$. Clearly this will only work if the point cloud is sufficiently close. Usually ``sufficiently close'' is interpreted as a bound on the Hausdorff distance between $A$ and $S$  (the Hausdorff distance between $A$ and $S$, denoted $d_H(A,S)$, is the smallest $r\geq 0$ such that $S\subseteq A_r$ and $A \subseteq S_r$).  Much of the earlier theory assumes that this Hausdorff distance is less than some measure of geometrical or topological feature size of the shape and show  the output is correct. We now survey some of this development.

The medial axis of a compact set $A$ is the set of points $p$ in the ambient space for which there is more than one point in $A$ which is closest to $p$. In Figure $1$ (a) and (b) the medial axes are the dashed lines. The reach of $A$ is the minimum distance between points in $A$ and points in the medial axis of $A$. It can be thought of as the minimum local feature size. The reach in (a) is $0$ and the reach in (b) is $c$. Sampling conditions based on reach include those found in \cite{NSW1} which consider smooth manifolds in $\R^n$.  Smooth submanifolds have positive reach but a wedge, for instance, does not. 

To deal with a larger class of sets Chazal, Cohen-Steiner and Lieutier in \cite{compactcrit} introduced the notion of $\mu$-reach. A point is $\mu$-critical when the norm of the gradient of the distance function at that point is less than or equal to $\mu$. In section $3$ we elaborate on a geometric description. In brief, a point $p$ is a $\cos \theta$\footnote{$\theta \in [0,\pi/2]$}-critical point of the distance function to $A$ if all the points in $A$ that are closest to $p$ cannot be contained in any cone emanating from $p$ with an angle less than $\theta$. In particular $0$-critical points are critical points of the distance function and every point on the medial axis is a $\mu$-critical point for some $\mu<1$. The $\mu$-reach of a set is the supremum of $r>0$ such that the $r$ offset does not contain any $\mu$-critical points. In  \cite{compactcrit} and \cite{ripsvscech} sampling conditions are given in terms of the $\mu$-reach. 

Another important feature size is the weak feature size. The weak feature size is the infimum of the positive critical values of the distance function from $A$. This has several advantages. Firstly, it can mean a significant improvement on the bounds such as in the case of ``hairy'' objects. Secondly, it can be applied to a larger class of compact sets. Every semi-algebraic set has positive weak feature size \cite{fu1985tubular}. This follows from the fact that the distance function from a semi-algebraic set has only finitely many critical values. Instead of making our bounds in terms of $\mu$-reach, we will only require the absence of $\mu$-critical points in an annular region of $A$ along with a bound on the weak feature size. 

\begin{figure}
\begin{center}
{\scalebox{1}{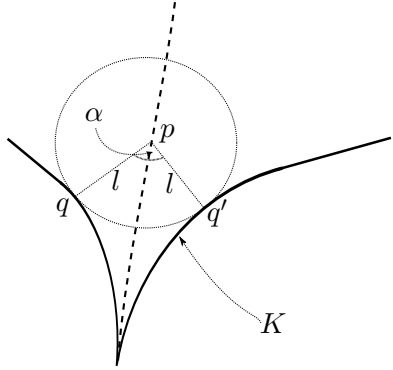}} 
\caption{The medial axes is the dashed line.  The weak feature size is $\infty$. The closest points in $K$ to $p$ are marked by $q$ and $q'$. Since $\angle(qpq')=\alpha$ we $p$ is $\cos(\alpha)$-critical. There are no $\cos \alpha$-critical points whose distance to the cusp greater than $l$.  $K$ has a $\mu$-reach of $0$ for all $\mu>0$.}
\end{center}
\end{figure}
\begin{figure}
\begin{center}
{\scalebox{1}{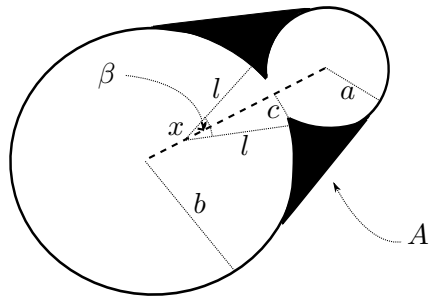}}
    \caption{ The medial axes is the dashed line. The weak feature size is $c$. The radii of the two circles, $a$ and $b$, are also critical values of the distance function. The $\mu$-reach is $c$ for all $\mu\in [0,1)$.  $x$ is a $\cos \beta$-critical point but  there are no $\cos\beta$ critical points in the annular region $\{x: l<d(x,A)<b\}$. }
 \end{center}
\end{figure}


Various feature sizes are illustrated in Figures $1$ and $2$. In Figure $1$, $p$ has two points in $K$ closest to $p$ marked by $q$ and $q'$. Since the angle between the geodesics from $p$ to $q$ and $p$ to $q'$ is $\alpha$ we conclude that $p$ is $\cos \theta$ critical for all $\theta<\alpha$ (or equivalently $\cos \theta >\cos \alpha$). There are no $\cos \alpha$-critical points whose distance to the cusp greater than $l$. Along the medial axis traveling toward the cusp we have $\mu$-critical points with $\mu$ tending to zero. This example show how a set can have a $\mu$-reach of $0$ for all $\mu>0$ and yet have a positive weak feature size (which in this example is infinity).

Now consider Figure $2$. The weak feature size is $c$. The radii of the two circles, $a$ and $b$, are also critical values of the distance function. The $\mu$-reach is $c$ for all $\mu\in [0,1)$. Now $x$ is a $\cos \beta$-critical point. Since $l>a$ we can say that there are no $\cos\beta$ critical points in the annular region $\{x: l<d(x,A)<b\}$. 
%
%
%

One limitation to any of these reconstruction theorems is the requirement of knowing geometric properties of the unknown object we are trying to reconstruct. A shift in perspective can overcome this limitation by considering the geometric properties of the point cloud, which we do know, and can hence prove sufficient conditions for an offset of the original set to deformation retract to an offset of the point cloud. We know the point cloud and hence we know the $\mu$-critical values of its distance function.
Theoretical guarantees are given in \cite{compactcrit} for when suitable homotopies exist by considering the $\mu$-reach of an offset of the point cloud. Unfortunately there is usually a significant number of small critical values of the distance function to the point cloud. This means the starting offset beyond which $\mu$-reach is considered is significant.  Our approach, which only considers the existence of $\mu$-critical points in a annular region, thus gains a significant advantage.

We note that previous reconstruction theories have been restricted to the case where the ambient space is Euclidean. Another contribution of this paper is to allow the ambient space to be any manifold whose curvature is bounded from below, thus answering an open question asked in \cite{compactcrit}. Although we focus on the important case of non-negatively curved manifolds we explore a paradigm of reconstruction which can be applied to manifolds with curvature bounded from below by some $\kappa<0$ with analogous, albeit messier, results. Examples of manifolds with nowhere negative curvature are Stiefel and Grassmannian manifolds. These examples are important because there are many applications where data naturally lies on these manifolds such as in dynamic textures \cite{dynamic}, face recognition \cite{facerecognition}, gait recognition \cite{gait} and affine shape analysis and image analysis \cite{affineshape}. 

Even when restricted to the case where we use $\mu$-reach in Euclidean space we still improve on the previous results whenever $\mu \leq 0.945$. The main theorem of this paper is as follows.

\begin{thm*}
Let $\mu\in (0,1)$, $r>0$. Let $\M$ be a smooth manifold with nowhere negative curvature such that every point has an injectivity radius greater than $r$.
Let $L$ a compact subset with $d_H(K,L) <\delta$. Suppose that there are no $\mu$-critical points in $K_{[r-\delta, r -\delta + 2\delta/\mu]}$ and $(4+\mu^2)\delta < \mu^2 r$. Then $L_r$ deformation retracts to $K_{r-\delta}$.
\end{thm*}
 
By considering both the set $A$ and its sampling point cloud $S$ in the role of $K$ we can reexpress this theorem as a sampling condition. It is a sampling condition as it gives a bound on the Hausdorff distance between the compact set we wish to reconstruct and its sampling point cloud.  

\begin{thm*}
Let $\mu\in (0,1)$, $r>0$. Let $A$ be a compact subset of a smooth manifold $\M$ with nowhere negative curvature such that the injectivity radius of every point in $\M$ is greater than $r$ and $A_r$ is homotopic to $A$. Let $S$ be a (potentially noisy) finite point cloud of $A$ (i.e. a finite set of points). 
Suppose that either
\begin{enumerate}
\item[(i)] there are no $\mu$-critical points of the distance function from $A$ in 
$\{x \in \M : d(x,A)\in [a, b]\}$

or
\item [(ii)] there are no $\mu$-critical points of the distance function from $S$ in $\{x\in \M : d(x,S) \in [a,b]\}$.
\end{enumerate}
Then $S_r$ is homotopic to $A$ whenever
$d_H(S,A)\leq \min\left\{r-a, \frac{b\mu -r\mu}{4-\mu}, \frac{\mu^2r}{4+\mu^2} \right\}.$
\end{thm*}
 
Note that if  $\text{wfs}(A) \geq a$ then there exists a deformation retraction from $A_a$ to $A_r$ for all $0<r<a$\cite{Grove}.

%
The new ingredient in our approach is the study of cone fields which are generalizations of not necessarily continuous unit vector fields where we attach  a closed ball in the unit tangent sphere, a ``cone'', to each point in the manifold. More precisely, at each point $x$ we chose a unit tangent vector $w_x$ and an angle $\beta_x$ and then we take the cone at $x$ to be the set of unit tangent vectors whose angle with $w_x$ is at most $\beta_x$. We denote this cone at $x$ by $C(w_x, \beta_x)$ and call it acute if $\beta_x$ is acute. A cone field is a choice of cone at each point. In section two we study cone fields, defining upper and lower semicontinuous cone fields and show that acute lower semicontinuous cone fields admit smooth vector fields. 

Of particular interest for our reconstruction theorem is the minimal acute $r$-spanning cone fields. The minimal acute $r$-spanning cone for $K$ from the point $x$, if it exists, is the cone  $C(w_x, \beta_x)$ where 
$$\exp_x\{tv: t\in [0,r],v\in C(w_x,\beta_x\}\supseteq K_\delta \cap B(x,r)$$ and $\beta_x$ is acute and minimal. We can then define its complementary cone to be $C(w_x, \pi/2-\beta_x)$. We will care whether the minimal acute $r$-spanning cone field (defined pointwise) for $K_\delta$ exists over the annular region $K_{[r-\delta,r+\delta]}:=\{x\in \M: d_K(x)\in [r-\delta, r+\delta]\}$. In section two we show if the complementary cone field to the $r$-spanning cone field admits a smooth vector field then the flow of this vector field produces a deformation retract from $L_r$ to $K_{r-\delta}$ when $d_H(K,L)\leq \delta$. Since the complementary cone of an upper semicontinuous cone field is lower semicontinuous, the problem of reconstruction is thus reduced to finding sufficient conditions for an acute $r$-spanning cone field of $K_\delta$ to exist over $K_{[r-\delta,r+\delta]}$ and that this minimal $r$-spanning cone field is upper semicontinuous.

We find sufficient conditions for the existence of an acute $r$-spanning cone field via the stability of $\mu$-critical points. A stability result of $\mu$-critical points when the ambient space is Euclidean is proved in \cite{compactcrit}. We prove a generalization of this result for when the curvature of the ambient space is bounded from below. The key to the proof is Toponogov's theorem about triangle comparison. It is worth observing that although $\mu$-reach is not stable under Hausdorff distance\footnote{In particular, for any compact set $K$ and any bound on Hausdorff distance $\delta>0$ there is a compact set $L$ with zero $\mu$-reach such that $d_H(K,L)<\delta$} we do have some stability of the absence of $\mu$-critical points within of annular regions.

The author thanks her advisor Shmuel Weinberger for helpful conversations and the anonymous referees for their very helpful constructive criticism.

\section{Cone fields}
One way to build a deformation retraction from $Y$ to $A$ is to construct a smooth vector field on $Y-A$ such that the vectors always point towards $A$ and never out of $Y$. More generally there may be some local condition such that if vectors in some smooth vector field satisfy it then the flow of the vector field has some desirable property. This leads us to the definition of cone fields which give a ball of acceptable unit vectors at each point. We will first rigorously define cone fields and then explore a sufficient condition for them to admit a smooth vector field. 

To define cone fields we must recall some differential geometry. A useful reference is \cite{leecurve}. Throughout $(\M, g)$ is a smooth $n$-dimensional manifold without boundary. The unit tangent bundle of a manifold $(\M, g)$, denoted by $UT\M$, is the unit sphere bundle for the tangent bundle $T\M$. It is a fiber bundle over $\M$ whose fiber at each point is the unit sphere in the tangent plane;
$$ UT\M := \coprod_{x \in M} \left\{ v \in T_{x} \M : g_x(v,v) = 1 \right\},$$
where $T_x\M$ denotes the tangent space to $\M$ at $x$. Elements of $UT\M$ are pairs $(x, v)$, where $x$ is some point of the manifold and $v$ is some tangent direction (of unit length) to the manifold at $x$. 

The exponential map at $x$ is a map from the tangent space $T_x \M$ to $\M$. For any $v \in T_x\M$, a tangent vector to $\M$ at $x$, there is a unique geodesic $\gamma_v$ satisfying $\gamma_v(0) = x$ with initial tangent vector $\gamma'_v(0) = v$. This uses the fact that geodesics travel at a constant speed. The exponential map at $x$ is defined by $\exp_x(v) = \gamma_v(1)$. The injectivity radius at a point $x$ is the radius of the largest ball on which the exponential map at $x$ is a diffeomorphism. Normal coordinates at a point $x$ are a local coordinate system in a neighborhood of $x$ obtained by applying the exponential map to the tangent space at $x$. 

Consider the $(n-1)$-dimensional unit sphere, $S^{n-1}$, lying inside $\R^n$ with a metric induced from this embedding. Denote by $C(w,\beta)$ the closed ball in $S^{n-1}$ centered at $w$ with radius $\beta$. We can view $C(w,\beta)$ as the intersection of $S^{n-1}$ with a particular infinite cone:
$$C(w, \beta) = S^{n-1} \cap \{ v\in \R^n\backslash \{0\} : \angle (v,w) \leq \beta\}.$$ 
We say  $C(w, \beta)$ is \emph{acute} if $\beta$ is acute.

We can equip the tangent bundle and the unit tangent bundle over a manifold with a Riemannian metric induced by the Levi-Civita connection.  Given a path $\gamma$ we can consider the linear isomorphism $\Gamma(\gamma)^t_s:T_{\gamma(s)}\M \to T_{\gamma(t)}\M$  induced by parallel transport along $\gamma$. This map is an isometry and so it sends the unit sphere to the unit sphere. We can then define a metric on $UT\M$ as follows. If $\gamma$ is a geodesic, $|s-t|$ small, $v\in T_{\gamma(s)}\M$, and $w\in T_{\gamma(t)}\M$ we define $$d_{UT\M}(v,w)^2 = (t-s)^2 + d_{UT_{\gamma(t)}\M}(w,\Gamma(\gamma)^t_s(v))^2.$$ 
Here we equip $UT_{\gamma(t)}\M$ with the usual metric on $S^{n-1}$.
 
In the case where $\M=\R^n$ is Euclidean space then $ \Gamma(\gamma)^t_s$ is just the identity map (we can of identifying tangent spaces by translation of the base point) and the metric on $UT\M$ is the same as that on the product space $\R^n\times S^{n-1}$.
 
Let $d_K$ denote the distance function from $K$. The Hausdorff distance between two compact sets $K$ and $L$ is denoted $d_H(K,L)$ and is defined by
$$ d_H(K,L) := \max \{\sup_{x\in K} d_L(x), \sup_{y\in L} d_K(y)\}.$$
Alternatively it is the smallest $r\geq0$ such that $K \subseteq L_r$ and $L \subseteq K_r$.

Denote by $F$ the fibre bundle over $\M$ where each fibre over $x \in \M$ is the space of non-empty closed balls in the unit tangent sphere at $x$.  $F$ has a natural metric induced from the Hausdorff metric on compact subsets of $UT\M$. A \emph{cone field} over a subset $U \subseteq \M$ is a section of $F$ restricted to $U$. A cone field is continuous if it is continuous as a section. As a set, we can write a cone field over $U$ as $\{ (x, C(w_x, \beta_x)): x\in U\}$ where $w_x \in UT_x\M$ and $\beta_x \in [0,\pi].$

%
%
%
%
%

One important observation is that if we take the parallel transport of a cone we again have a cone. More precisely if $\Gamma(\gamma)^t_s$ is the linear isomorphism induced by parallel transport along $\gamma$ then
$$\Gamma(\gamma)^t_s(C(w_{\gamma(s)}, \beta_{\gamma(s)}))=C(\Gamma(\gamma)^t_s(w_{\gamma(s)}),\beta_{\gamma(s)}).$$
We can consider vectors inside the cone at a point $x$. We say a vector field $X:=\{(x,v_x): x \in U, v_x \in T_x \M \}$ is \emph{subordinate} to the cone field $W = \{(x, C(w_x,\beta_x))\}$ if $v_x$  always lies in $C(w_x,\beta_x)$. We will call a vector field \emph{strictly subordinate} if the vector at each point lies in the interior of the cone.

Define the \emph{complementary cone} of $C(w,\beta)$ to be $C(w, \pi/2-\beta)$. Given a cone field where the cone at each point is acute we can construct the \emph{complementary cone field} pointwise. From the triangle inequality on the unit tangent sphere we obtain the following useful lemma. 

\begin{lem}\label{lem:comp} 
Let $v$ be a vector in some acute cone $C$ and $v'$ a vector strictly inside the complementary cone to $C$. Then $\angle (v,v') <\pi/2$.
\end{lem}

We will show the existence of smooth vector fields which are subordinate to particular cone fields. We will define a class of cone fields for which the existence of subordinate Lipschitz vector fields is guaranteed. Let $W= \{(x, C(w_x, \beta_x))\}$ be a cone field over $U$ with $\beta_x>0$ for all $x\in U$. Completely analogous to real-valued functions we say $W$ is \emph{upper semicontinuous} if for every $x\in U$ and every $\epsilon>0$ there is a $\delta>0$ such that for every unit speed geodesic $\gamma$ with $\gamma(0)=x$ we have
$$\Gamma(\gamma)^t_0(C(w_x, \beta_x+\epsilon)) \supseteq C(w_{\gamma(t)}, \beta_{\gamma(t)})$$ for all $t\in (0,\delta)$.
We say $W$ is \emph{lower semicontinuous} if for every $x\in U$ and every $\epsilon \in (0,\beta_x)$ there is a $\delta>0$ such that for every unit speed geodesic $\gamma$ with $\gamma(0)=x$ we have
$$\Gamma(\gamma)^t_0(C(w_x, \beta_x-\epsilon)) \subseteq C(w_{\gamma(t)}, \beta_{\gamma(t)})$$ for all $t\in (0,\delta)$.

Analogous the real-valued function case, one can prove that a cone field is continuous if and only if it is both upper and lower semicontinuous. There is a useful relationship between upper and lower semicontinuous cone fields which is analogous to the fact that the negative of a upper continuous function is lower semicontinuous and vice versa.

\begin{lem}\label{lem:upperlower}
An acute cone field is upper semicontinuous if and only if its complementary cone field is lower semicontinuous.
\end{lem}
\begin{proof}
The proof follows from the observation that $$\Gamma(\gamma)^t_0(C(w_x, \beta_x+\epsilon)) \supseteq C(w_{\gamma(t)}, \beta_{\gamma(t)})$$ if and only if $\angle(\Gamma(\gamma)^t_0(w_x), w_{\gamma(t)})\leq  \beta_x+\epsilon - \beta_{\gamma(t)}=(\pi/2-\beta_{\gamma(t)})-((\pi/2-\beta_x)-\epsilon)$ if and only if
$$\Gamma(\gamma)^t_0(C(w_x, (\pi/2-\beta_x)-\epsilon)) \subseteq C(w_{\gamma(t)}, \pi/2-\beta_{\gamma(t)}).$$
\end{proof}

\begin{prop}\label{prop:lowersemicont}
Let  $W=\{(x, C(w_x, \beta_x))\}$ be a lower semicontinuous cone field over $U\subset \M$ with $\beta_x>0$ for all $x\in U$. Then there exists a smooth unit vector field strictly subordinate to $W$.
 \end{prop}

\begin{proof}
Since $W$ is assumed to be lower semicontinuous, for each $x$ there exists a $\delta_x>0$ such that for every unit speed geodesic $\gamma$ with $\gamma(0)=x$ we have
$$\Gamma(\gamma)^t_0(C(w_x, \beta_x/2)) \subseteq C(w_{\gamma(t)}, \beta_{\gamma(t)})$$ for all $t\in (0,\delta)$.
This means that the vector field over $B(x, \delta_x)$ constructed by parallel transport of $w_x$ is strictly subordinate to $W|_{B(x, \delta_x)}$. Let $X_x$ denote this vector field. By construction $X_x$ is a smooth unit vector field on $B(x,\delta_x)$.

The set of open balls $\{B(x,\delta_x):x\in U\}$ cover $U$ and since $\M$ is paracompact there is a locally finite subcover $\{B(x_i, \delta_i)\}$ and a partition of unity $\{\rho_i\}$ subordinate to this cover. This means that there is a collection of functions $\rho_i:U \to [0,1]$ of smooth functions such that $\supp (\rho_i)\subseteq B(x_i, \delta_i)$ for each $i$ and for each $y\in U$ we have $\sum_i \rho_i(y)=1$.

Let $X=\sum_i\rho_iX_{x_i}$. This is defined over all of $U$ as the $X_{x_i}$ are all defined over the support of the corresponding $\rho_i$. It is smooth because the $X_{x_i}$ and the $\rho_i$ are all smooth. 

At each point $y \in U$ the vector $X(y)$, can be written as the sum $\sum_j a_j v_j$ where $\angle (v_j, w_y) < \beta_y<\pi/2$ and the $a_j \in [0,1]$ for each $j$ and $\sum_j a_j=1$. This implies that $\angle (\sum_j a_j v_j, w_y) <\beta_y$ and that $X(y)$ is nonzero. Thus we can conclude that $\frac{X(y)}{\|X(y)\|} \in W(y)$. 

Since $X$ is nowhere vanishing we can rescale the vectors at each point to construct a smooth vector field $\hat{X}$ of unit vectors. This vector field $\hat{X}$ is strictly subordinate to $W$. 
\end{proof}

We are interested in cone fields that reflect local geometric properties of the distance function to a set. We call $\gamma$ a \emph{segment} from $x\notin K$ to $K$ is if $\gamma$ is a distance achieving path from $x$ to $K$. If $\M$ is Euclidean then these segments are straight lines. Observe that on an arbitrary manifold there can be more than one segment connecting $x$ to the same $y \in K$. 

We say $C(w,\beta)$ is an \emph{$r$-spanning cone} for $K$ if 
$$\{\exp_x(tv): t\in[0,r], v\in C(w,\beta)\}\supseteq K\cap \overline{B(x,r)} \neq \emptyset.$$ We say $C(w,\beta)$ is a \emph{minimal} if whenever $C(w', \beta')$  is also an $r$-spanning cone then $\beta\leq\beta'$. We can see that when an acute minimal $r$-spanning cone exists then it is unique. 

\begin{lem}\label{lem:retract}
Let $0<\delta<r/2$, and let  $K, L$ be compact subsets of a manifold $\M$ with $d_H(K,L)\leq \delta$. Suppose that there exists an acute $r$-spanning cone field $W$ over $K_{[r-\delta, r+\delta]}$ for $K_\delta$. Let $W'$ be the complementary cone field to $W$. If $X$ is a smooth vector field strictly subordinate to $W'$, then $X$ induces a deformation retraction from $L_r$ to $K_{r-\delta}$.
\end{lem}
\begin{proof}
Since $X$ is a smooth vector field it has a unique smooth integral flow (a standard result, for example \cite{irwin1980}). The idea is to follow this flow from each point in $L_r$ until it reaches $K_{r-\delta}$. 

Denote the acute $r$-spanning cone field $W$ by $\{(x,C(w_x, \beta_x))\}$. From Lemma \ref{lem:comp} we know that any vector sitting strictly inside $C(w_x, \frac{\pi}{2}-\beta_x)$ forms an acute angle when paired with angle vector in $C(w_x, \beta_x)$. 

Let $X$ be a vector field strictly subordinate to $W'$. Let $x \in L_r \cap  K_{[r-\delta, r+\delta]}$, and let $v$ be the vector at $x$ in X.  Now $x \in B(y,r)$ for some $y \in L$. Let $\gamma_x^y$ denote a geodesic from $x$ to $y$ of length at most $r$, and ${\gamma'}_x^y(0)$ its tangent vector at $x$. By construction, ${\gamma'}_x^y(0) \in C(w_x, \beta_x)$ and hence it forms an acute angle with $v$. This means that their images form an acute angle in the normal coordinates  given by the exponential map at $x$.

Now consider the normal coordinates given by  the exponential map at $y$. In these coordinates $\gamma_x^y$ is a radial straight line emitted from the origin. Gauss's lemma (see \cite{leecurve}) tells us that the angle between $\gamma_v$  and $\gamma_x^y$ in the normal coordinates at $y$ is acute if and only if the angle between them in the normal coordinates  based at $x$ is acute. We have already shown that this second angle is acute. This means that in the normal coordinates given by the exponential map at $y$, $\gamma_v$ must remain inside  $B(0,r)$ for some positive amount of time. Since this is true all $x$ we know that the integral flow does not leave $L_r$. 

The integral flow of $X$ is always traveling towards $K$ as it lies in $W'$. For each $x\in K_{[r-\delta,r+\delta]}$, let $\lambda_x$ be the rate at which the integral flow of $X$ at $x$ is traveling towards $K$. Since $\lambda_x >0$ for all $x \in K_{[r-\delta,r+\delta]}$ and $K_{[r-\delta,r+\delta]}$ is compact there is some $\lambda>0$ which forms a lower bound on how fast  the integral flow of $X$ travels towards $X$. 

Construct the deformation retraction from $L_r$ to $K_{r-\delta}$ by following each point along the flow of $X$ until it reaches $K_{r-\delta}$ and then remaining stationary. The uniform lower bound on how fast the integral flow of $X$ travels towards $K$ combined with the observation that every point in $L_r$ is at most $2\delta$ from $K_{r-\delta}$, tells us that in a finite amount of time every point in $L_r$ will be sent to one in $K_{r-\delta}$.
\end{proof}

\section{Stability of $\mu$-critical points}

We want to study the gradient vector fields for distance functions from compact subsets of a general manifold $(\M, g)$. This can be thought of as the obvious generalization of the gradient vector fields for distance functions from compact subsets of Euclidean space (as studied in \cite{compactcrit} and \cite{medialaxis}). 


Recall that $\gamma$ is a segment from $x\notin K$ to $K$ is if $\gamma$ is a distance achieving path from $x$ to $K$. The point $x$ is a called a \emph{critical point} of the distance function from $K$ if, for all non-zero $v \in T_x \M$, there exists a segment $\gamma$ from $x$ to $K$ such that $\angle(\gamma'(0), v) \leq \pi/2$. Equivalently, if $\Gamma(x) := \{y \in K : d_K(x) = d(x,y)\}$, then $x$ is a critical point if and only if $0$ lies in the convex hull of $\exp_x^{-1}\Gamma(x)\cap \overline{B(0,d_K(x))}$. 

We need to construct the gradient vector field so that it vanishes at critical points of the distance function. For all non-critical points we can consider the minimal spanning cone $C(w_x, \beta_x)$ for $\Gamma(x)$ from $x$ of length $d_K(x)$. We set 
$$\nabla_K(x) : = -\cos(\beta_x)w_x$$
whenever $x$ is not critical. Observe that $ \nabla_{K_a}(x) = \nabla_{K} (x)$ whenever $d_K(x) >a$. For $\mu \in \R$, we call $x$ \emph{$\mu$-critical} if $\|\nabla_K(x)\| \leq \mu$. A point is $0$-critical exactly when it is a critical point for the distance function. 

It is easy to verify that these definitions agree with those given in \cite{compactcrit} when $\M$ is Euclidean.

We will want to prove a generalization of the stability result in \cite{compactcrit} where the ambient space is a manifold with curvature bounded from below in a suitable neighborhood of the compact subset under study. By appropriate scaling it is sufficient to consider the cases where the curvature is bounded from below by $1$, $0$ or $-1$. 

\begin{lem}\label{lem:toponogov}
Let $K\subset \M$ be a compact subset of a manifold $(\M, g)$ for which the curvature  on $K_{2\alpha}$ is bounded from below by $\kappa$. Let $x\in K_\alpha \backslash K$ be a $\mu$-critical point of $d_K$. Then for any $y \in K_{2\alpha}$ we have
\begin{align*}
\cos d_K(y)  \geq &\cos d(y,x) \cos d_K(x)- \sin d(y,x)\sin d_K(x) \mu  &\text{if }\kappa=1\\
 d_K(y)^2 \leq &d_K(x)^2 + d(x,y)^2 + 2 d(y,x)d_K(x) \mu &\text{if }\kappa=0\\
 \cosh d_K(y) \leq & \cosh d(y,x) \cosh d_K(x)+ \sinh d(y,x)\sinh d_K(x) \mu &\text{if }\kappa=-1.
  \end{align*}
\end{lem}
  
\begin{proof}
Let $\theta \in (0,\pi/2]$ such that $\cos \theta=\mu$. Let $\Gamma(x)= \{z\in K: d_K(z)=d(z,x)\}$ and set
$$\widehat{\Gamma(x)}:=\{z/\|z\| : z \in \exp_x^{-1}(\Gamma(x))\cap \overline{B(0,d_K(x))}\}.$$
Fix $y\in K_{2\alpha}$ and choose $v\in T_x\M$ such that $\exp_x(d(x,y)v)=y$.

We want to show that there is some $z \in K$ and length achieving geodesics $\gamma_x^y$ and $\gamma_x^z$ such that $\angle (y,x,z) \leq \pi -\theta$ where $\angle (y,x,z)$ is the angle between $\gamma_x^y$ and $\gamma_x^z$. Suppose not. This means that no point of $\widehat{\Gamma(x)}$ lies in $C(v, \pi-\theta)$. Geometrically this means that $\widehat{\Gamma(x)}$ must lie in the interior of $C(-v, \theta)$ which is the complement of $C(v, \pi-\theta)$ in the sphere.

However this implies that the minimal spanning cone for $\Gamma(x)$ from $x$ of length $d_K(x)$ lies strictly inside $C(-v, \theta)$ and  hence $\|\nabla_K(x)\| > \cos \theta =\mu$. This contradicts the assumption that $x$ is a $\mu$-critical point (i.e. $ \| \nabla_K(x)\|\leq \mu$). Thus by contradiction, there is some point $z \in K$
and length achieving geodesics $\gamma_x^y$ and $\gamma_x^z$ such that $\angle (y,x,z) \leq \pi -\theta$. 

Let $\triangle_{x,y,z}$ be the geodesic triangle with $\gamma_x^y$ and $\gamma_x^z$ such that $\angle (y,x,z) \leq \pi -\theta$ which we have just shown must exist. Let $\tilde \triangle_{\tilde x,\tilde y,\tilde z}$ be the corresponding triangle in $\M (\kappa)$, the manifold with constant curvature $\kappa$, where the length of the sides are preserved.  Toponogov's theorem is a triangle comparison theorem which quantifies the assertion that a pair of geodesics emanating from the same point spread apart more slowly in a region of high curvature than they would in a region of low curvature. The details of this theorem and its proof can be found in \cite{CheegerEbin}. By taking the contrapositive of Toponogov's theorem we know $\angle (\tilde y, \tilde x, \tilde z) \leq \angle (y,x,z) \leq \pi -\theta$ and hence $\cos \angle (\tilde y, \tilde x, \tilde z) \geq - \mu.$ 

We finally substitute $d_K(y) \leq d(\tilde{y},\tilde{z})$, $d_K(x) = d(\tilde{x},\tilde{z})$ and $\cos \angle (\tilde y, \tilde x, \tilde z) \geq - \mu$ into the spherical, Euclidean and hyperbolic cosine rules respectively to obtain the desired inequalities.
%
\end{proof}

Our stability result will arise from comparing two opposing inequalities - one from the previous lemma alongside one coming from the following lemma which is Lemma 4.1 in \cite{Lyt06}. It is easy to check the definition of $\| \nabla_K(x)\|$ coincides with $\| \nabla_x f\|$ when $f=d_K$ and $X$ is a manifold.\footnote{ $\| \nabla_x f\|$ is the nonnegative number $\max\{0,\limsup_{y\to x} \frac{f(y)-f(x)}{d(y,x)}\}$. That this is $\| \nabla_K(x)\|$ follows from our geometric construction of $\nabla_K(x)$ and from the cosine rule.}

\begin{lem}[Lemma 4.1 in \cite{Lyt06}] \label{lem:Lyt}
Let $X$ be a metric space. Suppose $f : X \to \mathbb{R}$ is a locally Lipschitz map, $x \in X$, and $f(x) = 0$. For
$\mu, r > 0$, assume that the ball $\overline{B(x,r)}$ is complete and that $\|\nabla_z f\| \geq \mu$ for each $z$ with
$d(z, x) < r$ and $ f(z) \geq 0$. Then for each $0 <C<\mu$ there is a point $z \in X$ with
$d(z, x) \leq r$ and $f(z) = Cr$.
\end{lem}

%

The following proposition is a generalization of critical point stability theorem in \cite{compactcrit} where the ambient space can now be any manifold with non-negative curvature.

\begin{prop}\label{prop:stable}
Let $K,L$ be compact subsets of $\M$ with $d_H(K,L)\leq\delta$. Let $x$ be a $\mu$-critical point of $d_{K}$.  If $$C\geq\mu + 2\sqrt{\frac{\delta}{d_{K}(x)}}$$ and $K_{d_K(x) + 4\delta/(C-\mu)}$ has nowhere negative curvature, then there exists a $C$-critical point $y$ of $d_L$ with $d_L(y)\geq d_L(x)$ and $y\in \overline{B(x, 4\delta/(C-\mu))}$.
\end{prop}

\begin{proof}
We want to show that there is some $y$ such that $\|\nabla_{L}(y) \| \leq C$ and $d_{L}(y)\geq d_{L}(x)$ and $d(x,y)\leq 4\delta/(C-\mu)$. Suppose not. Then there is some $\tilde{\mu}>C$ such that  $\|\nabla_{L}(y) \| \geq \tilde{\mu}$ whenever $d_{L}(y)\geq d_{L}(x)$ and $d(x,y)\leq 4\delta/(C-\mu)$.

If $C\geq \mu+ 2\sqrt{\delta/d_{K}(x)}$ then $d_K(x) - 4\delta/(C-\mu)^2\geq 0$ and hence we can construct 
$$K':=K_{d_K(x) - 4\delta/(C-\mu)^2} \text{ and }L':=L_{d_K(x) - 4\delta/(C-\mu)^2}.$$ By construction $d_H(K',L') \leq d_H(K,L) \leq \delta$ and $d_{K'}(x)=4\delta/(C-\mu)^2$. Using $f = d_{L'} - d_{L'}(x)$ and $r=4\delta/(C-\mu)$ in Lemma \ref{lem:Lyt} we know that there exists a point $y \in B(x, 4\delta/(C-\mu))$ such that $f(y) = C 4\delta/(C-\mu)$ which means $d_{L'}(y)= d_{L'}(x) +  C 4\delta/(C-\mu)$. Using $d_H(K',L')\leq\delta$ we can show that 
\begin{align*}
d_{K'}(y)& \geq d_{L'}(y)-\delta
= d_{L'}(x) + C 4\delta/(C-\mu) -\delta \geq d_{K'}(x) +  C 4\delta/(C-\mu) -2\delta.
\end{align*}
Since $d_{K'}(x)=4\delta/(C-\mu)^2$ we conclude that
\begin{align}\label{eq:stable_ineq1}
 d_{K'}(y) \geq \frac{4\delta + C 4\delta(C-\mu) - 2 \delta(C-\mu)^2}{(C-\mu)^2}.
 \end{align}

At the same time, Lemma \ref{lem:toponogov} implies that  $d_{K'}(y)^2 \leq  d_{K'}(x)^2 + d(x,y)^2 + 2 d(y,x)d_{K'}(x)\mu$
and hence
\begin{align}\label{eq:stable_ineq2}
d_{K'}(y)^2 \leq \frac{16 \delta^2  + 16 \delta^2 (C-\mu)^2 + 32 \delta^2 \mu(C-\mu)}{(C-\mu)^4}
\end{align}

By combining \eqref{eq:stable_ineq1} and \eqref{eq:stable_ineq2} and performing some algebraic manipulation we obtain  
%
\begin{align}\label{eq:man}
1 + (C-\mu)^2 + 2(C-\mu)\mu \geq (1 + C (C-\mu) - (C-\mu)^2/2)^2.
\end{align} However algebraic manipulation of \eqref{eq:man} implies
%
%
%
$ 0\geq   (C-\mu)^2/4  + C\mu$ which is a contradiction.
\end{proof}

%
%

One of the problems with working with the $\mu$-reach is that it is not stable under Hausdorff distance. Indeed by the creation of an arbitrarily small cusp we know that for any compact subset $K$, and any $\delta>0$, there exists some compact subset $L$ with $d_H(K, L)<\delta$ whose $\mu$-reach is zero for all $\mu>0$. However by only considering $\mu$-critical points in an annular regions we can have stability results.

\begin{cor}
Let $K, L$ be compact subsets of a manifold with non-negative sectional curvature such that $d_H(K,L)\leq \delta$. Suppose that there are no $C$-critical points for $d_{L}$ in the annular region $L_{[a, b]}$.  If $$C\geq\mu + 2\sqrt{\frac{\delta}{a+\delta}}$$
then there are no $\mu$-critical points for $d_{K}$ in the annular region $K_{[a+\delta, b-4\delta/(C-\mu) - \delta]}.$ 
\end{cor}

\begin{proof}
If $x$ is a $\mu$-critical point with $d_K(x) \in [a+\delta, b-4\delta/(C-\mu) - \delta]$ then by Proposition \ref{prop:stable} there exists some $C$-critical point $y$ with $$d_L(y) \in  [d_L(x), d_L(x)+ 4\delta/(C-\mu)] \subset [d_K(x)-\delta, d_K(x)+ 4\delta/(C-\mu) + \delta] \subset
 [a,b]$$ which is a contradiction. 
\end{proof}

Analogous stability results should hold for the cases when $\kappa=1, -1.$ However, we will only later require the case when $\mu=0$ and so to significantly simplify calculations we restrict to this case.


\begin{prop}\label{prop:-1}
Let $K,L$ be compact subsets of $\M$ with $d_H(K,L)\leq\delta$.  Let $x$ be a critical point of $d_{K}$.
Suppose that the sectional curvature of $K_{2d_K(x)}$ is bounded from below by $\kappa=-1$. 

Then for all $C>0$ there exists a $C$-critical point $y$ of $d_{L}$ with $d_{L}(y)\geq d_{L}(x)$ and $d(x,y)\leq 4\delta/C$ whenever $$9\delta \leq 2\tanh (d_K(x))C^2.$$
\end{prop}
\begin{proof}
Let $C\in(0,1)$. We want to show that there is some point $y$ such that $\|\nabla_L (y)\| \leq C$ and $d_{L}(y)\geq d_{L}(x)$ and $d(x,y)\leq 4\delta/C$. Suppose not. Then there is some $\tilde{\mu}>C$ such that $\|\nabla_L (y)\| \geq \tilde{\mu}$ whenever $d_{L}(y)\geq d_{L}(x)$ and $d(x,y)\leq 4\delta/C$

Using $f = d_L - d_L(x)$ and $r=4\delta/C$ in Lemma \ref{lem:Lyt} we know that there exists a point $y \in B(x,4\delta/C)$ such that $f(y)=4\delta$; i.e. $d_L(y)= d_L(x) + 4\delta$. From $d_H(K,L)\leq\delta$ we know that $ d_K(y) \geq d_K(x) + 2\delta.$

At the same time, Lemma \ref{lem:toponogov} implies that $\cosh d_K(y) \leq \cosh d(y,x) \cosh d_K(x)$ and hence $\cosh(d_K(x) +  2 \delta) \leq \cosh(4\delta/C)\cosh d_K(x).$ 

Using the hyperbolic cosh sum formula and dividing through by $\cosh d_K(x)$, this can be rewritten as  
$$\cosh (4\delta/C) - \cosh( 2\delta) \geq  \tanh(d_K(x))\sinh(2\delta).$$
Our assumption that  $9\delta \leq 2\tanh (d_K(x))C^2$ implies that $4\delta/C\leq 8\tanh(d_K(x))C/9<1$. Now $$\cosh(t)-\cosh(2\delta)<\cosh(t)-1<\frac{9}{16}t^2$$ whenever $t\in (0,1)$ and $\sinh(t)>t$ for all $t$. This means that we can conclude that
$$\frac{9}{16}\left(\frac{4\delta}{C}\right)^2> \tanh(d_K(x))2\delta$$ and hence that $9\delta>2\tanh(d_K(x))C^2$. This contradicts our assumption of $\delta$ implying that there must exist a suitably nearby $C$-critical point. 
\end{proof}
 
\section{Reconstruction theorem}

Our reconstruction proof will involve finding sufficient conditions for the existence of useful cone fields. We will use the stability of $\mu$-critical points to show the existence of acute minimal spanning cones of $A_\delta$ from points in an annular region.

\begin{lem}\label{lem:local}
Let $\mu \in (0,1)$, $r>\delta>0$ and $\M$ be a manifold with nowhere negative curvature. Let $K\subset \M$ be a compact subset and $x \in K_{r+\delta}$. If there are no $\mu$-critical points of $d_K$ in $K_{ [d_K(x), d_K(x)+ 2(r- d_K(x) + \delta)/\mu]}$ and  
$$\delta \leq d_K(x) - \frac{4-\mu^2}{4+\mu^2}r$$
then there is an acute $r$-spanning cone  for $K_\delta$ from $x$.

\end{lem}
\begin{proof}
Suppose, by way of contradiction, that $0$ is in the convex hull of $(\exp_x^{-1} K_\delta) \cap \overline{B(0, r)}$. 
 
Set $\hat{K}:=(\exp_x^{-1}K_\delta) \cap \overline{B(0,r)}$.
Define the map 
\begin{align*}
\varphi: \overline{B(0,r)} &\to \partial B\left(0,\frac{r+(d_K(x)-\delta)}{2} \right)\\
z &\mapsto  \frac{r+(d_K(x)-\delta)}{2}\frac{z}{\|z\|}.
\end{align*}
and set $\hat{L}$ to be $\varphi(\hat{K})$. This construction is illustrated in Figure $2$. By construction $d_{\hat{L}}(x) = \frac{1}{2}(r+(d_K(x)-\delta))$ and $$d_H(\exp_x\hat L,\exp_x\hat K) \leq  \frac{1}{2}(r-(d_K(x)-\delta)).$$

\begin{figure}
\begin{center}
\def\svgwidth{10cm}
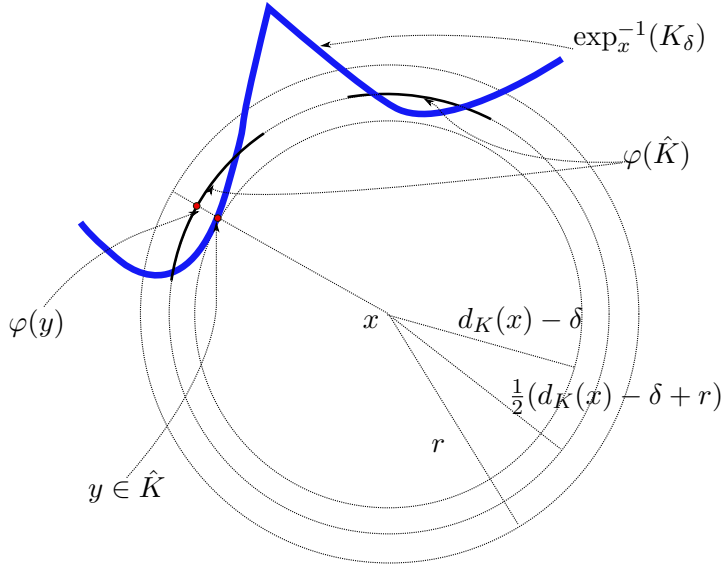

\caption{Construction of $\hat{L} = \varphi(\hat K)$}
\end{center}
\end{figure}

Set $L= \exp_x \hat{L} \cup\{y\in K: d(y,x)\geq r\}$. By construction $K = \exp_x \hat{K}  \cup\{y\in K: d(y,x)\geq r\}$. Since taking the union with both sets by the same set can only decrease the Hausdorff distance,  $d_H(K, L)\leq d_H(\exp_x \hat K,\exp_x\hat L)$. Also, by construction, $d_L(x) = d_{\hat{L}}(x) =  \frac{1}{2}(r+(d_K(x)-\delta))$.

Since $0$ is in the convex hull of $\hat{K}$ (by assumption), there exists $z_1, \ldots, z_m\in \hat{K}$ and $a_1, \ldots a_m>0$ such that $\Sigma_{i=1}^m a_i z_i =0$. However this gives
$$\sum_{i=1}^m a_i \frac{2\|z_i\|}{r +d_K(x)-\delta}\varphi(z_i)=0$$ and hence $0$ is in the convex hull of $\hat L$. Since all the points in $\hat L$ are equidistant from $0$, and hence all the points in $\exp_x \hat L$  are equidistant to $x$, we conclude that $x$ is a $0$-critical point of $d_{\exp_x\hat L}$. Adding points further away from $x$ does not affect this criticality and so $x$ is a $0$-critical point of $d_{L}$.

Our condition that $\delta \leq d_K(x) - r(4-\mu^2)/(4+\mu^2)$  can be rewritten as 
$$\frac{1}{2}\left(r - (d_K(x)-\delta)\right) < \frac{1}{2}\left(r +(d_K(x)-\delta)\right)\frac{\mu^2}{4}.$$
This implies that $d_H(K,L) < d_L(x) \mu^2/4$, or in other words $\mu\geq 0+2\sqrt{d_H(K,L)/d_L(x)}$. This enables us to apply Proposition \ref{prop:stable} to say that there exists a $\mu$-critical point $y$ of $d_K$ with 
$$d_K(y) \in [d_K(x), d_K(x)+ 4d_H(K,L)/\mu] =  [d_K(x), d_K(x)+ 2(r - d_A(x) + \delta)/\mu]  .$$ 
This contradicts our assumption about the absence of $\mu$-critical points in that annular region. 

\end{proof}

The process can be applied for the case when $\kappa=-1$ using Proposition \ref{prop:-1} instead of Proposition \ref{prop:stable}. This leads to the following lemma. 

\begin{lem}\label{lem:kappa-1}
Let $\mu \in (0,1)$, $r>\delta>0$ and $\M$ be a manifold with curvature bounded from below by $\kappa=-1$. Let $K\subset \M$ be a compact subset and $x \in K_{r+\delta}$. If there are no $\mu$-critical points of $d_K$ in $K_{ [d_K(x), d_K(x)+4\delta/\mu]}$ and  
$$9(r+ \delta-d_K(x))\leq4\tanh\left(\frac12(r-\delta+d_K(x)\right)\mu^2$$
there is an acute $r$-spanning cone  for $K_\delta$ from $x$.
\end{lem}

%
%
%

\begin{thm}\label{thm:big}
Let $\mu\in (0,1)$, $r>0$. Let $\M$ be a smooth manifold with nowhere negative curvature such that every point has an injectivity radius greater than $r$.
Let $K,L\subset \M$ be compact with $d_H(K,L) <\delta$. Suppose that there are no $\mu$-critical points in $K_{[r-\delta, r -\delta + 2\delta/\mu]}$ and $(4+\mu^2)\delta < \mu^2 r$. Then $L_r$ deformation retracts to $K_{r-\delta}$.
\end{thm}

\begin{proof}
Suppose that there are no $\mu$ critical point of $d_K$ in $K_{[r-\delta, r+\delta + 2\delta/\mu]}$. For each $x\in K[r-\delta, r+\delta]$ we have $d_K(x) \in [r-\delta, r+\delta]$ and hence 
$$[d_K(x), d_K(x) + 2(r-d_K(x)+\delta)/\mu] \subseteq [r-\delta, r+\delta+ 2\delta/\mu].$$
Lemma \ref{lem:local} tells us that there exists an acute $r$-spanning cone field $W = (x, C(w_x, \beta_x))$ over $K_{[r-\delta, r+\delta]}$ for $K_\delta$. 

By Lemma \ref{lem:retract}  it is sufficient to show that  there exists a smooth vector field strictly subordinate to the complementary cone $W'$ of $W$. Using Lemma \ref{lem:upperlower} and Proposition \ref{prop:lowersemicont} the theorem will follow if we can show $W$ is upper semicontinuous.

Let $x\in K_{[r-\delta, r+\delta]}$ and $\epsilon >0$. Since $C(w_x, \beta_x)$ is the minimal spanning cone for $K_\delta \cap \overline{ B(x,r)}$ of length $r$ and $d_{K_\delta}(x)\geq r-2\delta$ we have 
$$K_\delta \cap  \overline{B(x, r)}\subseteq \{\exp_x(tv): v\in C(w_x, \beta_x), t\in [r-2\delta,r]\}.$$  
This implies that there exists an $\alpha_0>0$ such that for all $\alpha<\alpha_0$ we have 
\begin{align}\label{eq:offseteq1}
\left(K_\delta\cap \overline{B(x,r)}\right)_\alpha\subseteq \{\exp_x(tv): v\in C(w_x, \beta_x+\epsilon/2), t\in [r-\delta-\alpha, r+\alpha]\}.
\end{align}

Define the sequence of compact sets $A_n, n \in \mathbb{N}, n\geq 1$ as follows. Let $y\in A_n$ if and only if $y\in K_\delta \cap \overline{B(x,1/n)}$ and there does not exist a path $\gamma:[0,1] \to K_\delta$ with $\gamma(0)=y$, $\gamma(1)\in K_\delta \cap \overline{B(x,r)}$ such that $d_{K_\delta \cap \overline{B(x,r)}}(\gamma(t))$ is strictly decreasing. The $A_n$ are compact because $K_{\delta}$ is closed. The $A_n$ are decreasing for the inclusion and $\cap_n A_n=\emptyset$. This implies that for some $n$, $A_n=\emptyset.$ Set $\tilde{r}:=\min\{1/n,\alpha\}$.

For this $\tilde{r}$ (using equation \eqref{eq:offseteq1} for the second inclusion) we have 
\begin{align}\label{eq:offseteq2}
K_\delta\cap \overline{B(x,r+\tilde{r})}& \subseteq \left( K_\delta \cap \overline{B(X,r)}\right)_{\tilde{r}}\notag\\ 
&\subseteq\{\exp_x(tv): v\in C(w_x, \beta_x+\epsilon/2), t\in [r-\delta-\tilde{r}, r+\tilde{r}]\}.
\end{align}
By recalling that $d_{K_\delta}(x) \geq r-2\delta$ we can refine \eqref{eq:offseteq2} to state 
\begin{align}\label{eq:cone contains}
K_\delta \cap  \overline{B(x, r+\tilde{r})}\subseteq \{\exp_x(tv): v\in C(w_x, \beta_x+\epsilon/2), t\in [r-2\delta,r+\tilde{r}]\}.
\end{align}

We may assume that $r+\tilde{r}$ is less than the injectivity radius by taking $\tilde{r}>0$ small enough.
Let $\Gamma_x^y$ denote the isometry between the tangent plane at $x$ to that at $y$ induced by parallel transport.   along the geodesic from $x$ to $y$. This is well defined when the distance between $x$ and $y$ is less than the injectivity radius. 

The function $F: B(x,\tilde{r})\times T_x\M \to \M$ defined by $(y,u) \mapsto \exp_y\circ \Gamma_x^y(u)$ is continuous in both $y$ and $u$. This implies that for each pair of compact sets $A \subseteq T_x\M$ and $L\subseteq \M$ such that $\int(F(x,A))\supset L$ there is a $\eta>0$ such that $F(y,A) \supseteq L$ whenever $d(x,y)<\eta$.

By taking $A=\{tv: v\in C(w_x, \beta_x+\epsilon/2), t\in [r-2\delta,r+\tilde{r}/2]\}$ and $L=\exp_x(\{tv: v\in C(w_x, \beta_x+\epsilon), t\in [r-2\delta-\tilde{r},r+\tilde{r}]\})$ we can conclude that there is a $\eta>0$ such that
\begin{align}\label{eq:cone 2}
\exp_x(\{tv&: v\in C(w_x, \beta_x+\epsilon/2), t\in [r-2\delta,r+\tilde{r}]\})\notag\\
&\subseteq \exp_y\Gamma_x^y( \{tv: v\in C(w_x, \beta_x+\epsilon), t\in [r-2\delta-\tilde{r},r+2\tilde{r}]\}
\end{align}
whenever $d(x,y)<\eta.$

We may assume that $\eta<\tilde{r}$. Combining \eqref{eq:cone contains} with the triangle inequality we know that for $d(x,y)<\eta$,
$$K_\delta\cap \overline{B(y,r)} \subseteq \exp_x \{(tv): v\in C(w_x, \beta_x+\epsilon/2), t\in [r-2\delta,r+\tilde{r}]\}.$$
We then use \eqref{eq:cone 2} to obtain
$$K_\delta\cap \overline{B(y,r)}\subseteq \exp_y\Gamma_x^y( \{tv: v\in C(w_x, \beta_x+\epsilon), t\in [r-2\delta-\tilde{r},r+2\tilde{r}]\}.$$
Now $\Gamma_x^y$ is an isometry and so the intersection of both sides with $\overline{B(y,r)}$ produces the containment $$K_\delta\cap \overline{B(y,r)}\subseteq \exp_y\Gamma_x^y( \{tv: v\in C(w_x, \beta_x+\epsilon), t\in [0,r]\}$$ whenever $d(x,y)<\eta.$
 
 $C(w_y,\beta_y)$ is defined to be the minimal spanning cone of $K_\delta \cap \overline{B(y,r)}$ of length $r$ and so $\exp_y \{tv: v\in C(w_y, \beta_y), t\in [0,r]\}\subseteq \exp_y \Gamma_x^y \{tv:v\in C(w_x, \beta_x+\epsilon), t\in [0,r]\}$. From the assumption that $r$ is less than the injectivity radius we conclude $C(w_y,\beta_y) \subseteq \Gamma_x^y C(w_x,\beta_x + \epsilon)$.

\end{proof}
 
By doing the same process but using Lemma \ref{lem:kappa-1} instead of Lemma \ref{lem:local} we get the analogous theorem for when the ambient space has its sectional curvature bounded below by $-1$.

\begin{thm}\label{thm:bigkappa}
Let $\mu\in (0,1)$, $r>0$. Let $\M$ be a smooth manifold whose sectional curvature bounded below by $-1$ and with an injectivity radius   greater than $r$.
Let $K,L\subset \M$ be compact with $d_H(K,L) <\delta$. Suppose that there are no $\mu$-critical points in $K_{[r-\delta, r -\delta + 4\delta/\mu]}$ and $9\delta<2\tanh(r-\delta)\mu^2$. Then $L_r$ deformations retracts to $K_{r-\delta}$.

\end{thm}

\section{Applications to point cloud data}

In this section we now consider the situation where we have some unknown compact set $A$ which we are wanting to understand and we can sample $A$ to generate a (potentially noisy) point cloud of $A$ which we will denote by $S$.
Historically geometric conditions have been given on $A$ for when offsets of the $S$ and $A$ are homotopic. This is because $A$ is often assumed to have nice geometric properties whereas $S$, as a point cloud, has many critical points of its distance function nearby. The corresponding theorem produced using Theorem \ref{thm:big} and Theorem \ref{thm:bigkappa} with $K=A$ and $L=S$ is as follows.

\begin{cor}\label{cor:big}
Let $\mu\in (0,1)$, $r>0$. Let $\M$ be a smooth manifold with sectional curvature bounded by $\kappa$ and whose injectivity radius is greater than $r$. Let $A$ be a compact subset of $\M$ and $S$ be a (potentially noisy) point cloud of $A$. Suppose that there are no $\mu$-critical points in $A_{[a, b]}$. Then $S_r$ is homotopic to $A_{r-d_H(S,A)}$ whenever $$d_H(S,A) \leq \min\left\{r-a, \frac{b\mu -r\mu}{4-\mu}\right\}, \text{ and } d_H(S,A)< \frac{\mu^2r}{4+\mu^2} \text{ if } \kappa=0$$
or $$ d_H(S,A) \leq \min\left\{r-a, \frac{b\mu -r\mu}{4-\mu}\right\}, \text{ and } 9d_H(S,A)<2\tanh(r-d_H(S,A))\mu^2 \text{ if } \kappa=-1.$$

Furthermore if $A_{r-d_H(S,A)}$ deformation retracts to $A$ then $S_r$ deformation retracts to $A$.
\end{cor}

\begin{proof}
In order to apply Theorem \ref{thm:big} or Theorem \ref{thm:bigkappa} we need to make sure  that
$$[a,b] \supset[r-d_H(S,A), r -d_H(S,A) + 4d_H(S,A)/\mu]$$
and also that $d_H(S,A)< \frac{\mu^2r}{4+\mu^2}$ or $9d_H(S,A)<2\tanh(r-d_H(S,A))\mu^2$ respectively.
\end{proof}

Of general interest is finding the homotopy type of $A$ rather than $A_{r}$. However, a sufficient condition for $A_b$ to deformation retract to $A_a$ ($0<a<b$) is that there are no 0-critical points in $A_b\backslash A_a$ \cite{Grove}. It would be impossible from a point cloud to distinguish $A$ from $A_a$ for sufficiently small $a>0$. Furthermore, there are many shapes, such as hairy objects, for which many offsets have a deformation retract even if there are small $0$-critical values.

We now want to present a paradigm for finding sufficient conditions on point cloud data for reconstructing any compact subset, lying in any Riemannian manifold, which has positive weak feature size. The first observation we need is that for Corollary \ref{cor:big} it is sufficient to have lower bounds on the sectional curvature and the injectivity radius only for the points in $A_{6r}$ and $A_{3r}$ respectively. This is because no points outside this region are used in any of the proofs.  Since $A$ is compact there is some $r>0$ such that the injectivity radius of every point in $A_{3r}$ is greater than $r$. Reduce $r$ if necessary to ensure that $r<\operatorname{wfs}(A)$ where $\operatorname{wfs}(A)$ is the weak feature size of $A$ which we have assumed is positive. $A_{3r}$ is compact so there is some finite lower bound on the sectional curvature for points in $A_{3r}$. By rescaling the metric on the ambient manifold if necessary (and with it scaling $r$) we can assume that the lower bound on sectional curvature is $0$ or $-1$. This means we can apply Corollary \ref{cor:big}. It is clear a suitable $\mu$ and bound on $d_H(A,S)$ in the Corollary must exist. Because $r<\operatorname{wfs}(A)$ we can further state that the $S_r$ deformation retracts to $A$. This paradigm of reconstruction processes shows that what the ambient manifold is does not pose a theoretical barrier to the existence of reconstruction proofs. 

The homological feature size of a set $A$ is the infimum of the distances $\alpha>0$ such that $A_\alpha$ has a different homology to $A$. If we were only interested in reconstructing a set with the same homology as the original set it would be sufficient to do the above reconstruction process with replacing the weak feature size with the homological feature size. After applying Corrollary \ref{cor:big} to show $S_r$ is homotopic to $A_r$ we observe that since the homological feature size of $A$ is greater than $r$ then $A_r$ is homotopic to $A$.

An alternative approach, as pointed out in \cite{compactcrit}, is to consider geometric properties of $S$ (or in their case offsets of $S$) itself rather than $A$. We can take his approach because $S$ is a compact set and we do not require any smooth structure. This means we can also conclude another corollary with $K=S$ and $L=A$.

\begin{cor}
Let $\mu\in (0,1)$, $r>0$. Let $\M$ be a smooth manifold with sectional curvature bounded by $\kappa$ whose injectivity radius is greater than $r$. Let $A$ be a compact subset of $\M$ and $S$ be a (potentially noisy) point cloud of $A$. Suppose that there are no $\mu$-critical points in $S_{[a, b]}$. Then $S_r$ deformation retracts to $A_r$ whenever $$d_H(S,A) \leq \min\left\{r-a, \frac{b\mu -r\mu}{4-\mu}\right\}, \text{ and } d_H(S,A)< \frac{\mu^2r}{4+\mu^2} \text{ if } \kappa=0$$
or $$ d_H(S,A) \leq \min\left\{r-a, \frac{b\mu -r\mu}{4-\mu}\right\}, \text{ and } d_H(S,A)<\frac{2}{9}\tanh(r-d_H(S,A))\mu^2 \text{ if } \kappa=-1.$$
Furthermore if $A_r$ is homotopic to $A$ then $S_r$ is homotopic to $A$.
\end{cor}


%

When the ambient space is Euclidean, it is reasonable to want to compare our reconstruction process to previous ones in the literature. Since these have been quantified in terms of $\mu$-reach we can first compare the required Hausdorff bounds on $\delta:=d_H(A,S)$ where $A$ is a compact set with  $\mu$-reach $r_{\mu} > 0$ and $S$ is a point cloud. If we consider the limiting case when $r-\delta +4\delta/C< r_{\mu}$ we can apply our reconstruction theorem (in the case of $\kappa=0$) once 
$\delta/r_\mu < \mu^2/(4+4\mu).$ 
Notably this is an improvement on the bounds presented in \cite{compactcrit}, which is $\delta/r_\mu < \mu^2/(5\mu^2 +12)$,  for all $\mu$ and an improvement on the bounds in \cite{ripsvscech}, where it is $$\frac{\delta}{r_\mu} < \frac{-3\mu +3\mu^2 - 3 + \sqrt{-8\mu^2 +4\mu^3 +18\mu+2\mu^4 +9+\mu^6-4\mu^5}}{7\mu^2+22\mu + \mu^4 -4\mu^3+1},$$  for $\mu<0.945$.

One advantage of the approach of this paper is not having any requirements about the absence of $\mu$-critical points very close to $A$. A severe limitation of restricting to set with positive $\mu$-reach is the inability to cope with sets that have cusps. At cusps the the $\mu$-reach is zero for all values of $\mu >0$. The method used to overcome the shortfalls of $\mu$-reach is to consider offsets of the compact set. For a compact set $K$,  there are no $\mu$-critical points of $d_K$ in $K_{[a,b]}$ if and only if the $\mu$-reach of $K_a$ is at least $b-a$. This means we can compare different reconstruction theorems in terms of a lack of $\mu$ critical points in an annular region.

Let us assume that there are no $\mu$-critical points of $d_K$ in $K_{[a,b]}$. For our reconstruction process we need 
$$\delta<\min\left\{\frac{\mu(b-a)}{4},  \frac{\mu^2 b}{4+4\mu}\right\}.$$
Here we would use $r=b\frac{4+\mu^2}{4+4\mu}$. In comparison, for the reconstructions in \cite{compactcrit} would need $$\delta< \frac{(b-a)\mu^2}{5\mu^2 +12}$$ and the the reconstructions in
\cite{ripsvscech} we would need
$$\delta < (b-a) \frac{-3\mu +3\mu^2 - 3 + \sqrt{-8\mu^2 +4\mu^3 +18\mu+2\mu^4 +9+\mu^6-4\mu^5}}{7\mu^2+22\mu + \mu^4 -4\mu^3+1}$$
which is significantly worse when $b-a$ is small in comparison to $b$.

One possible future direction is to use these results to find suitable sampling conditions for when a compact set can be reconstructed. In particular, probabilistic results would be interesting.

\section{Index of Notation}

\begin{itemize}
\item[$\M$] is a smooth Riemanian manifold which forms the ambient space.
\item [$A$] is a compact subset of $\M$ which we desire to reconstruct.
\item[$S$] is a noisy point cloud sample of $A$.
\item[$\delta$] is a bound on the Hausdorff distance between two compacts sets.
\item[$UT\M$] is the unit tangent bundle of $\M$.
\item[$T_x\M$] is the tangent plane to $\M$ at the point $x$.
\item[$\gamma$] is a geodesic on $\M$ (usually unit speed and always constant speed).
\item[$x,y,z$] are points in $\M$.
\item[$\exp_x$] is the exponential map from the tangent plane at $x$ to $\M$. 
\item[$\Gamma(\gamma)$] is the isometry between tangent planes induced by parallel transport along $\gamma$. 
\item[$w,v$] are unit tangent vectors.
\item[$\beta, \theta$] are angles. We mainly care about acute angles.
\item[$C(w,\beta)$] is a cone. It is a ball in the unit tangent sphere at a point in $\M$ with center $w$ and radius $\beta$.
\item[$W$] is a cone field. Also denoted by $\{(x,C(w_x,\beta_x))\}$.
\item[$W'$] denotes the complementary cone field to $W$ when $W$ is an acute cone field. For $W$ above it is $\{(x, C(w_x, \pi/2-\beta_x))\}$. 
\item[$X$] is a vector field.
\item[$K,L$] are compact subsets of $\M$.
\item[$d_K$] is the distance function from $K$.
\item[$K_a$] is the $a$-offset of $K$. That is $\{ x\in \M: d_K(x)\leq a\}$.
\item[$K_{[a,b]}$] is the $[a,b]$ annulus of $K$. That is $\{x\in M: a\leq d_K(x) \leq b\}$.
\item[$\nabla_K$] is the gradient vector field for $d_K$.
\end{itemize}
\bibliographystyle{plain}	
\bibliography{mybibreconstruction}


\end{document}